\documentclass[lettersize,journal]{IEEEtrans}  

\usepackage{amsmath,amssymb,dsfont,mathtools,stmaryrd}
\usepackage{amsthm}
\usepackage{algorithm,algpseudocode}

\usepackage[isbn=false,doi=false,url=false]{biblatex}
\AtBeginBibliography{\footnotesize}
\usepackage[font=footnotesize]{caption}
\usepackage{graphicx,subcaption}
\usepackage{siunitx}
\usepackage{mhchem}
\usepackage{glossaries}
\newacronym{mpc}{MPC}{model predictive control}
\newacronym{bess}{BESS}{battery energy storage system}
\newacronym{adp}{ADP}{approximate dynamic programming}
\newacronym{dp}{DP}{dynamic programming}
\newacronym{spm}{SPM}{single particle model}
\newacronym{pdae}{PDAE}{partial differential-algebraic equation}
\newacronym{sei}{SEI}{solid-electrolyte interphase}
\newacronym{soc}{SoC}{state of charge}
\newacronym{soh}{SoH}{state of health}

\usepackage{enumitem}
\setitemize{noitemsep,topsep=0pt,parsep=0pt,partopsep=0pt,leftmargin=*}

\usepackage[colorlinks = true,
linkcolor = blue,
urlcolor  = blue,
citecolor = blue,
anchorcolor = blue]{hyperref}
\usepackage[nameinlink]{cleveref}

\crefname{equation}{}{}

\usepackage{tikz}
\usetikzlibrary{trees,positioning,decorations.pathreplacing,decorations.text}

\newcommand{\E}[2]{\mathbb{E}_{#1}\left[#2\right]}
\newtheorem{assumption}{Assumption}
\newtheorem{proposition}{Proposition}
\newtheorem{remark}{Remark}

\DeclareMathOperator*{\argmin}{arg\,min}

\begin{document}

\title{
Approximate~Dynamic~Programming~for~Degradation-aware Market~Participation~of~Battery~Energy Storage~Systems: Bridging~Market~and~Degradation~Timescales
}

\author{Flemming Holtorf and Sungho Shin
  \thanks{
    This work was supported by the MIT Energy Initiative Seed Grant Program and the MIT Research Support Committee Funds.

    F. Holtorf and S. Shin are with the Department of Chemical Engineering,
    Massachusetts Institute of Technology, Cambridge, MA, USA (e-mail: holtorf@mit.edu; sushin@mit.edu).
  }
}

\maketitle

\begin{abstract}
We present an approximate dynamic programming framework for designing degradation-aware market participation policies for battery energy storage systems. The approach employs a tailored value function approximation that reduces the state space to state of charge and battery health, while performing dynamic programming along a pseudo-time axis encoded by state of health. This formulation enables an offline/online computation split that separates long-term degradation dynamics (months to years) from short-term market dynamics (seconds to minutes)---a timescale mismatch that renders conventional predictive control and dynamic programming approaches computationally intractable. The main computational effort occurs offline, where the value function is approximated via coarse-grained backward induction along the health dimension. Online decisions then reduce to a real-time tractable one-step predictive control problem guided by the precomputed value function. This decoupling allows the integration of high-fidelity physics-informed degradation models without sacrificing real-time feasibility. Backtests on historical market data show that the resulting policy outperforms several benchmark strategies with optimized hyperparameters.
\end{abstract}

\section{Introduction}
While \glspl*{bess} have generated substantial revenue for investors in recent years, increasing storage deployment fundamentally alters the market conditions that initially created these revenue opportunities.
In particular, as storage penetration rises, the very mechanisms that enable profitability—energy arbitrage and ancillary services provision—tend to weaken.
Previous studies have shown that the arbitrage and ancillary service value of storage declines with higher penetration levels because additional storage capacity dampens price volatility and reduces scarcity events, thereby compressing spreads within these market segments \cite{brijsQuantifyingElectricityStorage2019,denholmPotentialBatteryEnergy2020,frazierStorageFuturesStudy2021}. 
Empirical observations from real markets reinforce this pattern. As ancillary service markets approach saturation and wholesale price spreads narrow, reported annual revenues for 2-hour \gls*{bess} assets have fallen by nearly 50\% in major markets, including ERCOT, CAISO, and the UK \cite{DecliningCostsShifting2025}. 
These trends indicate that margins in BESS business models are increasingly under pressure, making it essential to \emph{extract maximum lifetime value through optimized operation}.

To maximize their lifetime value, the market participation policy for \glspl*{bess} must strike a delicate balance between immediate market rewards and the cumulative impact of usage on battery \gls*{soh}.
For example, aggressive operation that exploits short-term price fluctuations may accelerate battery degradation, eroding long-term profitability.
Similarly, a conservative strategy that prioritizes battery health may miss lucrative market opportunities (e.g., scarcity events), leading to suboptimal returns. Navigating this trade-off is further compounded by the following factors:
\begin{enumerate}
\item Market uncertainty: Prices and grid conditions evolve rapidly and partially unpredictably~\cite{cramer2023multivariate}. In addition, a significant portion of annual battery revenues is typically concentrated within a small number of highly volatile periods; for instance, two-hour batteries in ERCOT and Australia generated nearly half of their total revenues during the top \SI{10}{\percent} of days \cite{DecliningCostsShifting2025}.   
\item Complex degradation dynamics: Battery aging is governed by intricate electrochemical processes whose rates depend nonlinearly on operating conditions and usage patterns~\cite{o2022lithium,dufek2022battery,geslin2025dynamic}. High-fidelity models capturing the lithium-ion battery degradation mechanisms at the individual cell level, such as solid-electrolyte interphase growth and lithium plating, are available in the literature~\cite{reniers2019review,smith2017multiphase,doyle1993modeling}, but these models are characterized by a large intrinsic state-space dimension and their integration into real-time control frameworks is hindered by their computational complexity.
\item Separation of timescales: Battery \gls*{soc} and ancillary services/real-time energy market dynamics occur on the order of seconds to hours, while degradation unfolds over months or years, presenting a massive separation of timescales.
\end{enumerate}

These factors make the design of degradation-aware participation policies for \glspl*{bess} in wholesale electricity markets a stochastic, nonlinear, and multiscale problem in nature. Capturing long-term degradation with sufficiently high-fidelity physico-chemical battery models leads to a simultaneous explosion in the number of necessary decision epochs and state-space dimension, rendering off-the-shelve computational frameworks, such as \gls*{mpc}, computationally intractable. In addition, though \gls*{dp} provides a theoretically sound framework for the sequential decision-making under uncertainty required for optimal market participation of \glspl*{bess}, its application to degradation-aware control is hindered by the delayed reward structure: operational actions yield immediate returns (realized within minutes to hours), whereas costs incurred by degradation manifest much later (over years).
Furthermore, \gls*{adp} based on neural value function approximators and reinforcement learning methods---which rely on short-horizon feedback between actions and rewards---struggle to resolve this delay effectively and require exceedingly costly, long-horizon rollouts for policy iteration. This motivates the following question:
\begin{quote}
  \emph{How can we bridge the timescale gap between short-term market fluctuations and long-term degradation behavior in a computationally tractable market participation policy for \glspl*{bess}?}
\end{quote}

In this paper, we propose a tailored \gls*{adp} framework that turns the separation of timescales from curse to advantage.
By performing the \gls*{dp} induction along a pseudo-time axis defined by the monotonically decreasing battery \gls*{soh},
our formulation naturally decouples long-term degradation from fast market and grid dynamics, enabling a tractable offline and online computation split.
In the offline phase, we approximate the value function via coarse-grained backward induction along the health dimension using high-fidelity physicochemical battery models.
To keep the state dimension tractable, we perform a state-space reduction by projecting the internal battery state onto the \gls*{soh} and \gls*{soc} dimensions. In the online phase, market participation decisions are made with a lightweight, one-step \gls*{mpc} problem guided by the precomputed value function.
In this way, we can incorporate accurate, physics-informed degradation models without compromising real-time computational feasibility.
Backtests on historical market data demonstrate that our framework outperforms a range of benchmark strategies while maintaining a computational burden that is tolerable for real-time deployment.

\paragraph*{Related work}
Prior work on degradation-aware operation of grid batteries mostly adopt \emph{optimization-embeddable} aging surrogates---typically convex or piecewise-linear cycle-aging costs---to retain MILP/LP structure and the associated practical tractability. In contrast, physics-based electrochemical models are rarely used directly to decide market participation due to their large state dimension and stiff dynamics.
A widely used approach is to represent degradation as an explicit \emph{operational cost} that depends on cycle characteristics (e.g., depth of discharge) and sometimes rate effects. Xu \emph{et al.}~\cite{xuFactoringCycleAging2018} propose a piecewise-linear cycle-aging cost that approximates the underlying cycle-aging mechanism and can be incorporated into standard market clearing/dispatch formulations. Their key modeling step is to translate cycling into a marginal aging cost curve and then embed that curve into an optimization problem for participation in energy and reserve markets~\cite{xuFactoringCycleAging2018}. In a related direction, Padmanabhan \emph{et al.}~\cite{padmanabhanBatteryEnergyStorage2020} develop a \gls*{bess} operational cost model that explicitly accounts for degradation, with cost components parameterized by depth of discharge and discharge rate, and then derive a bid/offer structure for co-optimized energy and spinning reserve markets that internalizes this degradation-dependent operating cost in the market participation problem~\cite{padmanabhanBatteryEnergyStorage2020}. Foggo and Yu~\cite{foggoImprovedBatteryStorage2018} revisit lifetime valuation under cycle degradation and propose an approximate (co-optimizable) degradation model to reduce the value loss caused by cycling; the degradation surrogate is constructed from empirical aging data, primarily based on depth of discharge, to remain computationally convenient for co-optimization with operational decisions over long horizons~\cite{foggoImprovedBatteryStorage2018}. 
Sorourifar \emph{et al.}~\cite{sorourifarIntegratedMultiscaleDesign2020} propose a multiscale linear programming framework that simultaneously optimizes sizing, replacement, and market participation over horizons ranging from minutes to years while explicitly representing irreversible capacity loss due to degradation; the resulting formulation reaches very large scale (millions of variables and constraints) but remains solvable due to a deliberate choice of a linear degradation model~\cite{sorourifarIntegratedMultiscaleDesign2020}; a simplified mileage-based model is used to capture degradation.

\paragraph*{Contributions}
We present an \gls*{adp} framework that exploits the separation of timescales between degradation and market/grid dynamics to offload the computational burden of policy optimization to an offline phase, enabling the incorporation of high-fidelity, physics-informed degradation models without compromising real-time feasibility. Existing degradation-aware market participation and dispatch frameworks predominantly rely on (piecewise-)linear cycle-aging costs or empirical stress-factor models to retain tractability~\cite{xuFactoringCycleAging2018,padmanabhanBatteryEnergyStorage2020,foggoImprovedBatteryStorage2018,antoniadou-plytariaMarketBasedEnergyManagement2021,sorourifarIntegratedMultiscaleDesign2020,yongAgeDependentBatteryEnergy2025}. In contrast, high-fidelity physics-informed degradation models (e.g., porous-electrode or single particle models with explicit solid electrolyte interphase (SEI) growth) remain difficult to integrate into market participation because they substantially increase the state dimension and require fine time resolution. Our framework targets precisely this barrier by exploiting the separation of timescales within a tailored \gls*{adp} formulation. We demonstrate that with our framework, we can incorporate accurate, physics-informed degradation models into the policy design without compromising real-time computational feasibility. We validate the efficacy of our framework through backtests on historical market data, showing that it outperforms a range of benchmarks (at least 10\% improvement in cumulative returns) while maintaining computational efficiency suitable for real-time deployment.

\paragraph*{Organization}
The remainder of this paper is organized as follows. \Cref{sec:market_model} formulates a model for the optimal participation of degrading battery assets in uncertain electricity markets. \Cref{sec:dp} reviews the recursive \gls*{dp} solution for this model problem and lays the foundation for our \gls*{adp} framework, which is introduced in \Cref{sec:adp}, including both the offline value function approximation and the online policy implementation. \Cref{sec:experiments} compares the performance of our framework against baseline heuristics in backtests on historical market data, and \Cref{sec:conclusion} offers concluding remarks.

\paragraph*{Notation}
We assume throughout that all variables are real vector-valued and finite-dimensional unless stated otherwise. For brevity, we use the shorthand notation $g(x) \leq 0$ for component-wise non-positivity of a vector-valued function $g:\mathbb{R}^n \to \mathbb{R}^m$ at $x$. Furthermore, we denote the Jacobian of $g$ with respect to $x$ evaluated at $y$ by $\nabla_x g(y)$.
Finally, we use $\llbracket n \rrbracket$ and $\llbracket n\rrbracket_0$ as shorthands for the integer ranges $\lbrace 1, 2, \dots, n\rbrace$ and $\lbrace 0, 1, \dots, n\rbrace$, respectively.

\section{Problem~Setting: Degradation-Aware Market~Participation}\label{sec:market_model}
This section introduces a mathematical abstraction for the problem of optimal participation of degrading battery assets in uncertain electricity markets.
We first present a model for the relevant market signals and battery dynamics, followed by a description of the operational constraints, returns, and degradation dynamics that govern the decision-making problem.
For the sake of generality, we present the model in a way that is agnostic to the specific market segment (e.g., energy arbitrage, frequency regulation) and battery model (e.g., empirical, physics-based) under consideration.
\Cref{sec:experiments} then presents a concrete instantiation of the model for the case of frequency regulation market participation of lithium-ion batteries using a high-fidelity physics-based battery model.

\subsection{Market Signals and Battery Dynamics}
We begin by proposing a mathematical abstraction for the decision-making problem underpinning the participation of degrading \gls*{bess} in wholesale electricity markets. We assume that market participation decisions must be made in real-time at fixed bidding intervals that substantially exceed the timescale of relevant battery and grid dynamics. This assumption aligns with regulations in large electricity markets such as ERCOT, PJM, and CAISO, where participation decisions are made at increments of five to sixty minutes; in contrast, battery charging dynamics (e.g., current and voltage fluctuations) and relevant grid signals (e.g., frequency regulation signals) vary on the order of seconds. Formally and without loss of generality, we assume that bidding intervals are an integer $n$ multiple of the characteristic timescale of the relevant battery and grid dynamics.

Accordingly, we abstract the battery dynamics as a discrete-time dynamical system evolving on the nested discrete time-axis illustrated in \Cref{fig:timeaxis}:
\begin{align}
  \begin{aligned}
    x^k_{i} &= f(x^k_{i-1}, \xi_{i-1}^k, u^k),  && i \in \llbracket n \rrbracket,\  k \in \llbracket N \rrbracket_0\\
    x^{k}_0 &= x^{k-1}_n, &&  k \in \llbracket N \rrbracket.
  \end{aligned}
  \label{eq:model}
\end{align}
Here, $x^k_i\in\mathbb{R}^{d_x}$ and $\xi_{i}^k\in\mathbb{R}^{d_\xi}$ are the internal battery state and the relevant exogenous grid signals at the $i$\textsuperscript{th} time increment during the $k$\textsuperscript{th} bidding interval, respectively.
We introduce the shorthand notation $x^k = (x^k_0, x^k_1, \dots, x^k_n)$ and $\xi^k = (\xi_0^k, \xi_1^k, \dots, \xi_n^k)$ for notational brevity.
The exogeneous signals $(\xi^0, \xi^1, \dots, \xi^N)$ are considered random variables with known joint distribution $\Xi$. The market participation decisions $u^k \in \mathbb{R}^{d_u}$ represent power and ancillary service commitments for the $k$\textsuperscript{th} bidding interval throughout which they remain constant. The internal battery state, grid and market signals, and participation decisions determine the trajectory of the internal battery state as governed by the battery dynamics $f:\mathbb{R}^{d_x} \times \mathbb{R}^{d_\xi} \times \mathbb{R}^{d_u} \to \mathbb{R}^{d_x}$.

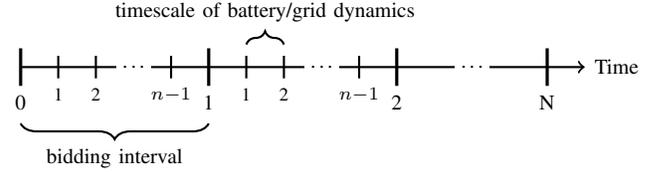
\begin{figure}[t]
    \centering
    \begin{tikzpicture}[font=\footnotesize,scale=1.0]
\def\delta{0.5}
\def\axisy{0} 
\def\braceYOffset{0.75} 
\def\fastBraceYOffset{0.3} 

\def\x{0}
\coordinate (slow0) at (\x, \axisy);

\draw[->, thick] (\x, \axisy) -- (7.5, \axisy) node[right] {Time};


\draw[very thick] (\x, -0.25) -- (\x, 0.25);
\node[below] at (\x, -0.25) {0};

\pgfmathsetmacro\x{\x + \delta}
\draw[thick] (\x, -0.15) -- (\x, 0.15);
\node[below,font=\scriptsize] at (\x, -0.15) {1};

\pgfmathsetmacro\x{\x + \delta}
\draw[thick] (\x, -0.15) -- (\x, 0.15);
\node[below,font=\scriptsize] at (\x, -0.15) {2};

\pgfmathsetmacro\x{\x + 2*\delta}
\draw[thick] (\x, -0.15) -- (\x, 0.15);
\node[below,font=\scriptsize] at (\x, -0.15) {$n{-}1$};
\node[fill=white, inner sep=1pt] at ({\x - \delta}, \axisy) {\,\dots};

\pgfmathsetmacro\x{\x + \delta}
\coordinate (slow1) at (\x, \axisy);
\draw[very thick] (\x, -0.25) -- (\x, 0.25);
\node[below] at (\x, -0.25) {1};

\pgfmathsetmacro\x{\x + \delta}
\coordinate (fast1b) at (\x, \axisy);
\draw[thick] (\x, -0.15) -- (\x, 0.15);
\node[below,font=\scriptsize] at (\x, -0.15) {1};

\pgfmathsetmacro\x{\x + \delta}
\coordinate (fast2b) at (\x, \axisy);
\draw[thick] (\x, -0.15) -- (\x, 0.15);
\node[below,font=\scriptsize] at (\x, -0.15) {2};

\pgfmathsetmacro\x{\x + 2*\delta}
\draw[thick] (\x, -0.15) -- (\x, 0.15);
\node[below,font=\scriptsize] at (\x, -0.15) {$n{-}1$};
\node[fill=white, inner sep=1pt] at ({\x - \delta}, \axisy) {\,\dots};

\pgfmathsetmacro\x{\x + \delta}
\draw[very thick] (\x, -0.25) -- (\x, 0.25);
\node[below] at (\x, -0.25) {2};

\pgfmathsetmacro\x{\x + 4*\delta}
\draw[very thick] (\x, -0.25) -- (\x, 0.25);
\node[below] at (\x, -0.25) {N};
\node[fill=white, inner sep=1pt] at ({\x - 2*\delta}, \axisy) {\,\dots};

\draw [decorate, thick, decoration={brace, amplitude=6pt}]
  ([yshift=-\braceYOffset cm]slow1) -- ([yshift=-\braceYOffset cm]slow0)
  node[midway, below=6pt] {bidding interval};

\draw [decorate, thick, decoration={brace, amplitude=5pt}]
  ([yshift=\fastBraceYOffset cm]fast1b) -- ([yshift=\fastBraceYOffset cm]fast2b)
  node[midway, above=5pt] {timescale of battery/grid dynamics};

\end{tikzpicture}
    \caption{Nested discrete time axis for market and battery/grid dynamics.}\label{fig:timeaxis}
\end{figure}

\subsection{Operational Constraints}
The safe and legal operation of \glspl*{bess} dictates compliance with a number of technological and market regulation constraints. Such constraints include, for instance, obeying maximum (dis)charging rates and avoidance of deep discharge or overcharging while ensuring to meet any made market commitments. In the following, we assume that these constraints can be cast as inequalities that involve jointly the internal battery state trajectory $x^k$ and participation decision $u^k$ during the $k$\textsuperscript{th} bidding interval; formally,
\begin{align}
    g(x^k, u^k) \leq 0, \ k \in \llbracket N\rrbracket_0, \label{eq:constraints}
\end{align}
where $g:\mathbb{R}^{(n+1)d_x} \times \mathbb{R}^{d_u} \to \mathbb{R}^{d_g}$ is a vector-valued function that encodes the relevant constraints. 

\subsection{Returns}
We assume that market participation in each bidding interval is remunerated by returns $r(u^k, \xi_0^k)$, where $r:\mathbb{R}^{d_u} \times \mathbb{R}^{d_\xi} \to \mathbb{R}$ is a function that encodes the market remuneration scheme and $\xi_0^k$ is the market and grid information available at the time of bidding. We assume that returns depend only on the commitments made $u^k$ and the prices of electricity and ancillary services at the time of bidding, i.e., $\xi_0^k$.

\begin{remark}
  In real electricity markets, returns depend technically not only on market and grid information that is available at the time of bidding but are subject to small uncertainties; for instance, the remuneration for committed regulation capacity is typically determined by not only the clearing price, but also the actual regulation signal during the bidding interval, which are not known exactly at the time of bidding. However, these uncertainties can be effectively marginalized out by using the expected returns (typically simply determined by expected prices), conditioned on the information available at the time of bidding, i.e., $r(u^k, \xi_0^k) = \E{\xi^k|\xi_0^k}{\tilde{r}(u^k, \xi^k)}$, where $\tilde{r}:\mathbb{R}^{d_u} \times \mathbb{R}^{(n+1)d_\xi} \to \mathbb{R}$ is a function that encodes the remuneration scheme based on the actual market and grid information during the bidding interval. Thus, the assumption that the returns depend only on the information available at the time of bidding is not restrictive and can be easily relaxed if necessary.
 
\end{remark}

\subsection{Battery Degradation}
The overall revenue potential of a \gls*{bess} is naturally tied to the number of bidding cycles it can participate in and is therefore affected by battery degradation. Moreover, as observed in experiments and predicted by mechanistic physico-chemical battery models, degradation rates and battery lifetimes vary dramatically depending on usage patterns~\cite{geslin2025dynamic, keil2016calendar,dufek2022battery}. A central question for revenue-maximizing \gls*{bess} management thus becomes how to optimally trade off momentary returns against diminishing future revenue potential due to usage-induced lifetime reduction.

To incorporate this aspect of \gls*{bess} management, we assume a battery reaches its end of life when its health---a function of the internal battery state typically measured in terms of remaining charge capacity---decays to a given threshold $h_{\min}$; formally, we define the bidding cycle at which end of life is reached as
\begin{align*}
  N_\text{EOL}
  = \max \lbrace k \in \mathbb{N}_0 : h(x_n^k) \geq h_{\min} \rbrace, 
\end{align*}
where $h:\mathbb{R}^{d_x} \to \mathbb{R}$ is a function that maps the internal battery state to a scalar health metric.
Note that $N_\text{EOL}$ depends on the trajectory of the internal battery state, which in turn depends on the market and grid signals and the participation decisions.

We make the following assumption on battery health decay.
\begin{assumption}\label{a:strict_decay}
  There exists a minimum decay rate $\epsilon > 0$ such that 
    \begin{align*}
        h \circ f(x, \xi, u)  \leq h(x) -\epsilon
    \end{align*}
    holds for all exogenous market signals $\xi$, along all state trajectories, and for all participation decisions admissible under \cref{eq:constraints}.
\end{assumption}
This assumption aligns both with the understanding of battery aging mechanisms on a microscopic level~\cite{hahn2018quantitative,o2022lithium} and empirical findings that support the notion calendar aging---irreversible battery health decay in the absence of usage~\cite{lam2025decade}.
Crucially, under \Cref{a:strict_decay}, the battery has finite life $N_\text{EOL} \leq \frac{h(x^0_0) - h_{\min}}{n \epsilon}$.

\section{Dynamic Programming Formulation}\label{sec:dp}

\subsection{Value Function and Optimal Policy}
The optimal trade-off between short- and long-term revenue potential is concisely encoded in the value function $V(x_0^0, \xi_0^0)$, which quantifies the expected remaining lifetime value of the \gls*{bess} conditioned on its current internal state $x_0^0$ and available market and grid information $\xi_0^0$. Formally, the value function is characterized by the policy optimization problem of the following form:
\begin{align}
    V(x^0_0, \xi_0^0) \coloneqq \sup_{\pi \in \Pi} \ \E{(x,\xi,u,N_\text{EOL}) \sim P_\pi}{\sum_{k=0}^{N_\text{EOL}} \gamma^{k} r(u^k,\xi^k_0)}. \label{eq:value_fxn}
\end{align}
This problem seeks an admissible participation policy $\pi$ that maximizes the expected (discounted) cumulative return until the battery's end of life. Here, $P_{\pi}$ denotes the joint distribution of the trajectories of process \cref{eq:model} and the conditional market uncertainty $\xi \sim \Xi | \xi_0^0 $ under the policy $u^k = \pi(x^k_0, \xi^k_0)$. The set of admissible participation policies $\Pi$ is defined implicitly by the operational constraints \cref{eq:constraints}, which must hold under any realization of the uncertain market signal $\xi \in \text{supp} \, \Xi |\xi_0^0$. The discount factor $\gamma \in [0,1]$ reflects the decision-maker's preference for immediate short-term revenue over delayed long-term returns. It is worth noting that strict discounting ($\gamma < 1$) is {\em not} required to ensure finiteness of the value function due to the finite battery life guaranteed by \Cref{a:strict_decay}. 
One can observe that $V(x^0_0,\xi_0^0) = 0$ if $h(x^0_0) \leq h_{\min}$; that is, the battery value vanishes when end of life is reached.

\subsection{Dynamic Programming Perspective}
The value function defined in \cref{eq:value_fxn} admits the recursive \gls*{dp} characterization:
\begin{align}\label{eq:dp}
  \begin{aligned}
    &V(x^0_0, \xi_0^0) = \sup_{u} \ r(u, \xi_0^0) + \gamma \E{\xi \sim \Xi |\xi_0^0 }{V(x_n(\xi), \xi_0^1)} \\
    &\text{s.t.} \ \begin{dcases}
      x_0(\xi) = x^0_0 \\
      x_{i}(\xi) = f(x_i(\xi), \xi^0_i, u), \ i \in \llbracket  n \rrbracket \\
      g(x(\xi),u) \leq 0 \\
      h(x_n({\xi})) \geq h_{\min}
    \end{dcases}, \ \xi \in \text{supp} \, \Xi |\xi_0^0,
  \end{aligned}
\end{align}
which states that the optimal value at the current state and market information can be obtained by optimizing over the current participation decision $u$ to maximize the immediate return plus the expected discounted value of the battery at the beginning following bidding interval. This battery state is determined by the battery dynamics and the realization of the uncertain market signal $\xi$.
The optimal policy $\pi^*$ can be derived from this characterization by selecting, for each state and market information, a participation decision that attains the supremum in \cref{eq:dp}.

\section{A Practical~\gls*{adp}~Framework~Exploiting Separation~of~Timescales}\label{sec:adp}

Although the \gls*{dp} recursion in \cref{eq:dp} can in principle be solved approximately with various \gls*{dp} algorithms, it remains computationally intractable for accurate physics-based battery degradation models. These mechanistic models describe the microscopic transport phenomena underlying batteries through coupled \glspl*{pdae}, resulting in state space dimensions far beyond those for which \gls*{dp} is practically viable.

Reinforcement learning methods that have proven remarkably effective for policy optimization in high-dimensional spaces~\cite{arulkumaran2017deep,wang2023recent} are likewise challenged by this application. Assessing degradation-induced trade-offs requires extremely long horizons~\cite{vezhnevets2017feudal,kulkarni2016hierarchical,sutton1999between,sutton1984temporal,dayan1992feudal}, while accurate battery and grid modeling demands fine temporal resolution. The combined requirement of fine time resolution for accurate simulation and extremely long look-ahead horizons renders policy roll-outs exceedingly expensive and gradient evaluation instable.  

To overcome these challenges, we propose an \gls*{adp} routine that avoids \gls*{dp} in a high-dimensional state space and explicitly exploits the separation between the timescales of the relevant operational battery and grid dynamics and battery degradation. Instead of performing backward induction along real time, we perform coarse-grained backward induction along a pseudo-time axis encoded by battery health. This leads to a tractable offline/online computation split: in the offline phase an approximate value function is constructed on a joint battery health and state of charge grid while the online phase determines market participation decisions via a real-time tractable one-step \gls*{mpc} problem guided by the precomputed value function proxy as terminal cost.

\subsection{State Space Reduction and Lifting}
We approximate the value function as depending only on \gls*{soc} $q_0$ and \gls*{soh} $h_0$ and hence seek a proxy $\widehat{V}(q_0,h_0,\xi_0^0)$ for the true value function $V(x_0^0,\xi_0^0)$; for the sake of brevity, we omit the explicit dependence of $q_0$ and $h_0$ on the full internal battery state $x_0^0$. This approximation is motivated by the fact that the combination of \gls*{soc} and \gls*{soh} provides an effective compression of the full internal battery state from the perspective of economic decision-making: \gls*{soc} determines short-term revenue potential by encoding immediately available charge and discharge capacity, while \gls*{soh} reflects long-term revenue potential as an indicator of remaining useful life.

The full internal battery state -- represented by microscopic quantities such as spatial concentration profiles -- evolves on shorter timescales and, due to the underactuated nature of batteries, within a more constrained set of dynamics than the derived summary quantities \gls*{soc} and \gls*{soh}. The economic impact of the fine-grained microscopic nuances is therefore only relevant for accurate prediction of the battery dynamics on horizons that are well resolved within a single bidding interval. As such, their impact is effectively captured by the one-step \gls*{mpc} problem used for policy implementation as discussed in \Cref{sec:policy}. 

To enable our computational scheme, we assume the existence of a lifting map $l(q_0,h_0)$ that maps \gls*{soc} and \gls*{soh} to an approximate but consistent full internal battery state $x_0^0$. Such lifting maps are standard tools for coarse-grained simulation of multiscale systems~\cite{kevrekidis2009equation} which has previously been applied to accelerate cycle aging simulations for high fidelity battery models in~\cite{sulzer2021accelerated}. In particular, we demonstrate in \Cref{sec:experiments} that for the widely adopted single particle model with degrading \gls*{sei} layer growth, such a lifting map arises naturally from standard approximations.

\subsection{Offline Phase: Coarse-grained Backward Induction Along The Battery Health Axis}
As battery health decays monotonically (cf. \Cref{a:strict_decay}), it admits the interpretation as a pseudo-time variable that naturally encodes the timescale of battery degradation. We exploit this feature to compute the value function proxy $\widehat{V}(q_0, h_0,\xi_0^0)$ in a recursive manner by applying backward induction along this battery health axis rather than real time. 

Recalling the \gls*{dp} recursion, point-wise evaluation of the value function proxy at $(q_0,h_0,\xi_0^0)$ can be approximated with the following \gls*{mpc} problem. 
\begin{align}
  &\widehat{V}(q_0,h_0,\xi_0^0) \coloneq \sup_{x,u} \, R(x, u,\xi_0^0) \tag{MPC($q_0, h_0, \xi_0^0; \varphi$)} \label{eq:mpc}\\
  &\text{s.t.} \begin{dcases}
    x_0(\xi) = l(q_0, h_0)  \\
    x_{i}(\xi) = f(x_i(\xi), \xi^0_i, u) \ i \in \llbracket n \rrbracket \\
    g(x(\xi),u) \leq 0 \\  
    h(x_n(\xi)) \geq h_{\min}
  \end{dcases}, \ \xi \in \text{supp}\,\Xi|\xi_0^0 \nonumber
\end{align}
where
\begin{align*}
  R(x, u, \xi_0^0) = r(u, \xi_0^0) +  \gamma \E{\xi \sim \Xi |\xi_0^0 }{\varphi (q(x_n(\xi)), h(x_n(\xi)), \xi)}.
\end{align*}
with a suitable smooth terminal cost surrogate $\varphi$ for the value function $V$.

It is important to emphasize that under sample average approximation \gls*{mpc} is a finite nonlinear program that is readily solved by off-the-shelf primal-dual interior point solvers~\cite{wachterImplementationInteriorpointFilter2006}. Moreover, under mild regularity conditions, solving \cref{eq:mpc} with a primal-dual solver not only enables point-wise evaluation of $\widehat{V}(q_0, h_0, \xi_0^0)$ but also of $\frac{\partial \widehat{V}}{\partial h}(q_0, h_0, \xi_0^0)$ via the following sensitivity result.
\begin{proposition}\label{prop:sensitivity}
  Let $(x^*,u^*,\lambda^*)$ be a primal-dual feasible point of \cref{eq:mpc} corresponding to the unique global optimum. Further, assume that $(x^*, u^*, \lambda^*)$ satisfies the strong second-order sufficient condition, linear independence constraint qualification, and strict complementary slackness~\cite[see Chapter 3 for detailed definitions]{bertsekas1999nonlinear}. Denote by $\lambda_{0}(\xi)$ the Lagrange multiplier for the constraints  $x_0(\xi) = l(q_0, h_0)$, for $\xi \in \textup{supp} \, \Xi | \xi_0^0$.  Then,
  \begin{align*}
    \frac{\partial \widehat{V}}{\partial h}(q_0, h_0, \xi_0^0) = - \sum_{\xi \in \textup{supp}\, \Xi | \xi_0^0} \lambda^*_{0}(\xi)^{\top} \nabla_h l(q_0,h_0).
  \end{align*}
\end{proposition}
\begin{proof}
  The result follows immediately from \cite[Proposition 3.3.3]{bertsekas1999nonlinear} and the chain rule. 
\end{proof}

The final ingredient for the offline phase of our \gls*{adp} framework is a tailored regression scheme that bridges the gap between point-wise evaluations of the value function proxy (and its health gradient) via \cref{eq:mpc} and the terminal cost $\varphi$ required for solving \cref{eq:mpc} in the first place. To break this circular dependency, assume for the moment that point-wise evaluations of the value function proxy and its health derivative are available at a given battery health $h_0$ for a finite \gls*{soc} grid $\mathcal{Q}$ and a finite set $\mathcal{M}$ of representative scenarios for the market uncertainty in a single bidding interval, i.e., 
\begin{align*}
    \left\lbrace \left( \widehat{V}(q,h,\xi),\frac{\partial \widehat{V}}{\partial h}(q, h, \xi_0) \right) : (q, h) \in \mathcal{Q}\times \mathcal{M}\right \rbrace.
\end{align*}
Then, we may approximate $\widehat{V}$ locally around $h_0$, for a fixed market uncertainty $\xi\in \mathcal{M}$, and for \gls*{soc} $q$ by $\varphi$ via
\begin{align*}
    \varphi(q,h,\xi) = \widehat{V}_{\theta_1(\xi)}(q) + {\rm d}\widehat{V}_{\theta_2(\xi)}(q) (h - h_0),
\end{align*}
where $\widehat{V}_{\theta_1(\xi)}$ and ${\rm d}\widehat{V}_{\theta_2(\xi)}$ are parametric function approximators determined via the regression problems 
\begin{align}
        \begin{dcases}
        \theta_1(\xi) \in \argmin_{\theta} \sum_{q \in \mathcal{Q}} \ell_0 \left( \widehat{V}_{\theta}(q),  {\widehat{V}(q,h_0,\xi)}\right)\\
        \theta_2(\xi) \in \argmin_{\theta} \sum_{q \in \mathcal{Q}} \ell_1 \left ( {\rm d} \widehat{V}_{\theta}(q), \frac{\partial \widehat{V}}{\partial h}(q,h_0,\xi) \right)    
        \end{dcases}, \xi \in \mathcal{M} \label{eq:regression}
\end{align}    
with suitable regression losses $\ell_0$ and $\ell_1$.

Putting all pieces together, a practical backward induction algorithm is obtained as follows. In addition to the \gls*{soc} grid $\mathcal{Q}$ and market uncertainty scenarios $\mathcal{M}$, consider a battery health grid $\mathcal{H} = (h_1, \dots, h_m)$ that covers the entire battery life with $h_{\min} < h_1 < \cdots < h_m$. Recalling that the battery value vanishes at $h_{\min}$, in the initial induction step, \cref{eq:mpc} is solved for battery health $h_1$ and all combinations of $(q_0, \xi_0^0) \in \mathcal{Q}\times \mathcal{M}$ with the terminal cost $\varphi \equiv 0$. This yields data $\lbrace \widehat{V}(h_1, q, \xi_0^0): (q, \xi) \in \mathcal{Q} \times \mathcal{M} \rbrace$ for the regression problems \cref{eq:regression}, which in turn furnish a local approximation of $\widehat{V}$ around $h_1$. Updating the terminal cost $\varphi$ with this approximation, we can proceed inductively along the \gls*{soh} grid $\mathcal{H}$. \Cref{alg:adp} summarizes this procedure in detail.

Finally, a few remarks are in order. First, we emphasize that the linear dependence of $\varphi$ on \gls*{soh} is a deliberate choice that explicitly reflects the pseudo-time nature of battery health and enables tractable backward induction in practice. In particular, the error in this approximation can be controlled directly by the size of the backward induction steps. Second, we note that while the computational cost of our \gls*{adp} is substantial, it can be significantly offset by exploiting the embarrassingly parallel nature of several substeps as indicated in \Cref{alg:adp}. Moreover, all computations can be done upfront and offline, hence do not impose limits on real-time tractability. 

\begin{algorithm}[t]
\caption{Value function approximator}\label{alg:adp}
\begin{algorithmic}
\State \textbf{Input:} Battery health grid $\mathcal{H} = (h_1, h_2, \dots, h_m)$ such that $h_{\min} < h_1 < \cdots < h_m$, \gls*{soc} grid $\mathcal{Q} = (q_{\min}, q_1, \dots, q_{\max})$, market scenarios $\mathcal{M}$, parametric function approximators $\widehat{V}_\theta$ and ${\rm d}\widehat{V}_{\theta}$, regression losses $\ell_0$ and $\ell_1$, discount factor $\gamma\in [0,1]$. 
\noindent\rule{\linewidth}{0.4pt}
\State \textbf{Output:} Value function approximation $\widehat{V}(q,h,\xi)$ for all $(q,h,\xi) \in \mathcal{Q}\times\mathcal{H}\times \mathcal{M}$.
\noindent\rule{\linewidth}{0.4pt}
\State Set $\varphi(q,h,\xi) \equiv 0$ 
\For{$h_0 \in \mathcal{H}$}
    \For{$(q_0,\xi^0) \in \mathcal{Q} \times \mathcal{M}$} \Comment{in parallel}
        \State Compute primal-dual feasible point $(x^*, u^*, \lambda^*)$ \\ \qquad \quad of \cref{eq:mpc} and set
        \begin{align*}
            &\widehat{V}(q_0,h_0,\xi_0^0) \leftarrow R(x^*, u^*, \xi_0^0)\\
            &\frac{\partial \widehat{V}}{\partial h}(q_0,h_0,\xi_0^0) \leftarrow - \sum_{\xi \in \textup{supp}\, \Xi |\xi_0^0 } \lambda^{*}_{0}(\xi)^\top \nabla_h l(q_0, h_0)
        \end{align*}
    \EndFor
    \For{$\xi \in \mathcal{M}$}\Comment{in parallel}
    \State Solve regression problems \cref{eq:regression} and set 
    \begin{align*}
        \varphi(q,h,\xi) \leftarrow \widehat{V}_{\theta_1(\xi)}(q) + {\rm d}\widehat{V}_{\theta_2(\xi)}(q)(h-h_0)
    \end{align*}
    \EndFor
\EndFor
\end{algorithmic}
\end{algorithm}

\subsection{Online Phase: One-step model predictive control}\label{sec:policy}
During real-time operation, we decide the market participation by solving \cref{eq:mpc} ahead of each bidding interval. The terminal cost $\varphi$ is constructed by the same regression scheme described in the previous section, approximating the value function around the health grid point closest to the current battery health. After sample average approximation, this reduces \cref{eq:mpc} again to a tractable nonlinear program that can be solved efficiently with standard primal-dual methods~\cite{wachterImplementationInteriorpointFilter2006}. The computational burden of long-horizon degradation planning is thus deferred entirely to the offline phase, while online decision-making reduces to a one-step \gls*{mpc} problem.

\section{Numerical Experiments}\label{sec:experiments}
\subsection{Battery Model}\label{sec:battery_model}
We assume that the \gls*{bess} is composed of a large number of identical lithium ion battery cells, each described by a \gls*{spm}~\cite{ning2004cycle}.
The \gls*{spm} abstracts both positive and negative battery electrodes as collections of identical spherical particles. It provides a middle ground between crude reservoir-based models and the more detailed Doyle-Fuller-Newman pseudo-two-dimensional model. Importantly, its complexity-accuracy trade-off has been shown to be particularly suitable for control and optimization applications~\cite{cao2020multiscale,perez2017optimal,li2018single} and allows straightforward integration of physicochemical battery degradation models~\cite{cao2020multiscale,li2018single}.  

The \gls*{spm} describes the transport phenomena governing the macroscopic behavior of lithium ion batteries in terms of a system of spatio-temporal \gls*{pdae}s. Charging and discharging are modeled by a lithium intercalation reaction that occurs at the electrode particle surfaces followed by radial Fickian-diffusive lithium transport within the particles. The governing transport equations for the concentration of lithium ions in electrode $i\in \lbrace +, - \rbrace$ are thus
\begin{align*}
    \begin{dcases}
        \frac{\partial c_i}{\partial t} + \frac{D_i}{r^2} \frac{\partial}{\partial r}\left( r^2 \frac{\partial c_i}{\partial r} \right)  = 0, \  (r,t) \in (0, R_i) \times (0, \overline{t}] \\
        - D_i \left. \frac{\partial c_i}{\partial r}\right|_{r=R_i} = j_i \text{ and } \left.\frac{\partial c_i}{\partial r}\right|_{r=0} = 0, \ t\in (0,\overline{t}] \\
        c_i(r, 0) = c_{i,0}(r), \ r \in [0,R_i],
    \end{dcases}
\end{align*}
where $c_i$ is the lithium ion concentration, $D_i$ is the diffusion coefficient, $R_i$ is the particle radius for electrode $i$, and $j_i$ is the lithium ion flux at the particle surface due to the intercalation reaction. The first equation describes radial diffusion of lithium ions within the electrode particles, while the second and third equations specify the boundary and initial conditions, respectively.

The lithium ion flux due to the intercalation reaction at the particle surface is assumed to obey Butler-Volmer kinetics, i.e., 
\begin{align*}
    j_i = k_i \sqrt{(c^{\max}_i - c_i(R_i,t)) c_i(R_i,t) c_e} \sinh \frac{RT}{2F} \eta_i,
\end{align*}
where $k_i$ is the reaction rate constant, $c_i^{\max}$ is the maximum lithium ion concentration in electrode $i$, $\eta_i$ is the electrochemical overpotential driving the intercalation reaction, $c_e$ is the electrolyte concentration, $T$ is the battery cell temperature, and $F$ and $R$ refer to the universal Faraday and ideal gas constants, respectively.
Here,  the intercalation reaction is solely driven by the electrochemical overpotential $\eta_i$.
The electrolyte concentration $c_e$ is assumed constant while the potential and concentration gradients in the electrolyte phase as well as any axial concentration gradients are neglected.
Moreover, the battery cell temperature $T$ is assumed constant at $\SI{25}{\celsius}$ in line with the typically ample cooling capacity of grid-scale \gls*{bess}.

As the primary mechanism for degradation, we consider gradual growth of the \gls*{sei} layer~\cite{ramadass2004development}. The \gls*{sei} layer grows via a side reaction at the particle surface of the negative electrode. The thickness of the formed film, $w_{\rm SEI}$, follows the dynamics
\begin{align*}
  \frac{{\rm d} w_{\rm SEI}}{\rm{d} t} = \frac{j_{\rm SEI}}{\rho_{\rm SEI}} , t \in (0,\overline{t}],\quad w_{\rm SEI}(0) = w_{\rm SEI,0},
\end{align*}
where $j_{\rm SEI}$ and $\rho_{\rm SEI}$ denote the ionic surface flux due to side reaction and the molar density of the \gls*{sei} layer. Since \gls*{sei} layer growth acts as a sink for lithium ions, it leads to capacity fade
\begin{align*}
  \frac{{\rm d} h}{{\rm d} t} = - F j_{\rm SEI}, \;t \in (0,\overline{t}],\quad \text{and} \quad        h(0) = h_0.
\end{align*}
The ionic surface flux for the \gls*{sei}-forming side reaction follows $j_{\rm SEI} = k_{\rm SEI} \exp\left(- \frac{F}{RT}\eta_{\rm SEI}\right)$.
Importantly, it is easy to see that $j_{\rm SEI} > 0$ and thus $\frac{{\rm d} h}{{\rm d}t}  < 0$, implying this model is predictive of calendar aging with capacity fade acting as a strictly monotonically decreasing pseudo-time for the control problem. 

The overpotentials driving intercalation and SEI formation reactions are given by
\begin{align*}
    &\eta_{+} = \phi_+ - \eta_{+}^{\rm OCV}\left(\frac{c_+(t, R_+)}{c_+^{\max}}\right) \\
    &\eta_{-} =  \phi_- + \Delta \phi_{{\rm SEI}} - \eta_{-}^{\rm OCV}\left(\frac{c_-(t, R_+)}{c_-^{\max}}\right) \\
    &\eta_{\rm SEI} = \phi_- + \Delta \phi_{\rm SEI} - \eta_{\rm SEI}^{\rm eq}
\end{align*}
where $\Delta \phi_{\rm SEI}$ is the potential drop across the \gls*{sei} layer
\begin{align*}
    \Delta \phi_{\rm SEI} = \left(R^0_{\rm SEI} + \frac{w_{SEI}}{\sigma_{\rm SEI}}\right) i_{\rm app}
\end{align*}
for an applied current with density $i_{\rm app}$ at the negative electrode. $R_{\rm SEI}^0 + \frac{w_{SEI}}{\sigma_{\rm SEI}}$ quantifies the thickness-dependent resistance of the \gls*{sei} layer. $\eta_-^{OCV}$ and $\eta_+^{OCV}$ are the open-circuit potentials for the negative and positive electrodes, respectively. Their functional form is given by empirical relationships fitted to experimental data~\cite{cao2020multiscale}. Similarly, $\eta^{\rm eq}_{\rm SEI}$ denotes the constant equilibrium potential for the \gls*{sei}-forming side reaction. 

Finally, the ionic fluxes must satisfy the closure condition
\begin{align*}
    \frac{i_{\rm app}}{F} = -(j_{\rm SEI} + j_{-}) = \frac{S_+}{S_-} j_{+},
\end{align*}
where $S_+/S_-$ denote the relative surface areas of the positive/negative electrodes. 

All model parameters are taken from \cite{forman2012genetic} and the open-circuit potentials from~\cite{cao2020multiscale}. The parameters are for A123 Systems' ANR26650M1 cells with \ce{LiFePO4} cathode as suitable for high-power applications such as grid-scale \gls*{bess}.

\subsection{Operational Constraints}
For safe operation, we impose bound constraints on the voltage across the cell $U=\phi_+ - \phi_-$ and the applied current $S_- i_{\rm app}$, i.e., $U \in [\SI{2.4}{\volt}, \SI{3.65}{\volt}]$ and $ S_- i_{\rm app} \in [-10\text{\,C}, 10\text{\,C}]$. To avoid overcharging and deep discharging, we constrain the \gls*{soc} $c_- / c_-^{\max}$ to remain in the healthy window $[0.3,0.9]$. 

\subsection{Numerical Approximation \& Lifting Map}
In order to apply \Cref{alg:adp}, the \gls*{pdae} system underpinning the \gls*{spm} must be approximated by an algebraic equation system as per \cref{eq:model}. To that end, we employ two distinct approximations.

First, we make a parabolic approximation to the radial concentration profiles in the electrode particles, which is known to introduce minimal error for moderate applied currents~\cite{subramanian2001approximate,subramanian2004boundary}. Under this approximation, the spatio-temporal \gls*{pdae} governing diffusive transport inside the particles reduces to differential algebraic relations governing the dynamics of the average and surface particle concentrations, i.e., 
\begin{align*}
    \begin{dcases}
        \frac{{\rm d} \bar{c}_i}{\rm{d}t} = \frac{3}{R_i} j_i, & t\in (0,\overline{t}]\\
        \bar{c}_i(0) = \bar{c}_{i,0}
    \end{dcases}\text{ and } c_{i}(R_i,\cdot) = \bar{c}_i - 5\frac{R_i}{D_i} j_i
\end{align*}
for $i \in \lbrace +,-\rbrace$.

Second, the resulting differential-algebraic equation system is approximated by a discrete-time dynamical system by $5$\textsuperscript{th} order-accurate Gauss-Radau collocation on finite elements in time. The time grid is chosen uniformly with \SI{10}{\second} increments, coinciding with the timescale of available grid frequency data. 

A convenient consequence of the employed parabolic approximation in the \gls*{spm} is that it gives rise to a straightforward lifting map for use in \Cref{alg:adp}. Under this approximation, the \gls*{soc} $q$ is entirely characterized by the average lithium concentration in the negative electrode particles, i.e., $q = \frac{\bar{c}_-}{c_-^{\rm tot}}$,
where the normalization $c_-^{\rm tot}$ derives from the total amount of cyclable lithium: $c_-^{\rm tot} = \bar{c}_- + \frac{V_+}{V_-} \bar{c}_+$.

Furthermore, the dynamics of the total amount of cyclable lithium, battery health, and \gls*{sei}-layer thickness are governed exclusively by the \gls*{sei}-forming molar flux acting as the sole sink for lithium ions,
\begin{align*}
    j_{\rm SEI} = \rho_{\rm SEI} \frac{{\rm d} w_{\rm SEI}}{{\rm d} t} = - \frac{1}{F} \frac{{\rm d} h}{{\rm d} t} = -\frac{1}{S_-} \frac{{\rm d}c^{\rm tot}_-}{{\rm d}t}. 
\end{align*}
It follows that, given a known reference state of battery health $h_{\rm ref}$, \gls*{sei} layer thickness $w_{\rm SEI, ref}$, and total amount of cyclable lithium $c^{\rm tot}_{-,{\rm ref}}$, knowledge of the \emph{current} \gls*{soc} $q$ and battery health $h$ can be lifted to the entire microscopic battery state via 
\begin{align*}
    l: (q,h) \mapsto \begin{dcases} c^{\rm tot}_- = c^{\rm tot}_{-, {\rm ref}}  + \frac{S_-}{F} (h-h_{\rm ref}), \\
    \bar{c}_- = q c^{\rm tot}_-, \\
    \bar{c}_+ = \left(1 - \frac{V_-}{V_+} q\right)c_{-}^{\rm tot} \\
    w_{\rm SEI} = w_{\rm SEI, ref} -\frac{1}{F \rho_{\rm SEI}} (h-h_{\rm ref})
    \end{dcases}
\end{align*}
and the algebraic relationships provided in the previous section. The reference point is arbitrary, but in most applications is conveniently derived from the nominal battery capacity. We finally wish to emphasize that this lifting incurs no errors beyond the parabolic approximation to the radial concentration profiles in the electrode particles. Moreover, this lifting map can be easily adapted to more complicated models such as the \gls*{spm} with electrolyte~\cite{moura2012battery} or the pseudo two-dimensional Doyle-Fuller-Newman model~\cite{doyle1993modeling}.

\subsection{Uncertainty Model}\label{sec:uncertainty_model}
We consider an empirical, scenario-based model for the conditional market uncertainty $\Xi|\xi_0^0$. Detailed uncertainty modeling is a grand challenge in its own right and hence beyond the scope of this contribution.
To probe the utility of the proposed \gls*{adp} framework, we assume that an accurate forecast of the frequency regulation signal is available for the imminent bidding interval. Realization of the market uncertainty for subsequent bidding intervals is assumed independent of prior uncertainty realizations. For each stage, we consider a finite number of uncertainty realizations $\mathcal{M}$.

For the numerical experiments presented in this section, the scenarios are constructed to match the empirical marginals of the joint distribution of electricity prices, regulation capacity prices, and the frequency regulation signal of the French electricity market for the calendar year 2021. The data was retrieved from the Réseau de Transport d'Électricité (frequency regulation capacity prices and grid frequency signal)~\cite{rte} and the European Network for Transmission System Operators for Electricity (day-ahead market prices)~\cite{entsoe}. We use day-ahead market prices as a proxy for real-time market prices since, to the best of our knowledge, there is no publicly available record of the latter.

To construct scenarios that yield accurate approximations for the expectations with respect to the empirical marginals, we choose the electricity and frequency regulation capacity price scenarios and their probabilities to coincide with Gauss-Legendre quadrature nodes and weights under the empirical inverse cumulative distribution function. The empirical marginals and resultant scenarios are shown in \Cref{fig:price_scenarios}. 

\begin{figure}[t]
    \centering
    \begin{subfigure}{\linewidth}
        \centering
        \includegraphics[width=0.8\linewidth]{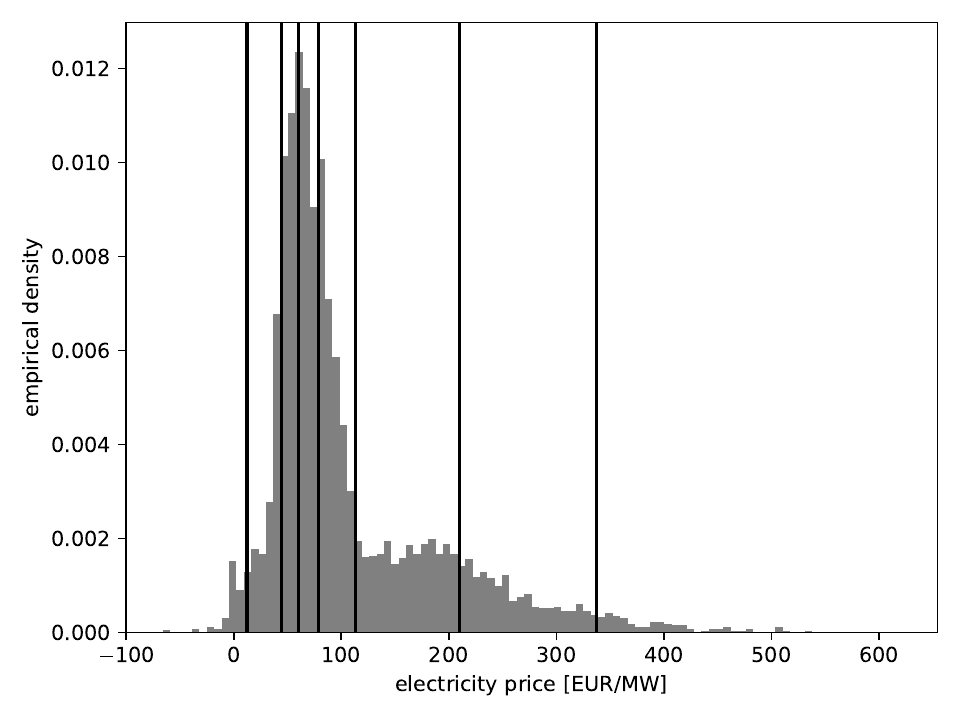}  
    \end{subfigure}
    \begin{subfigure}{\linewidth}
        \centering
        \includegraphics[width=0.8\linewidth]{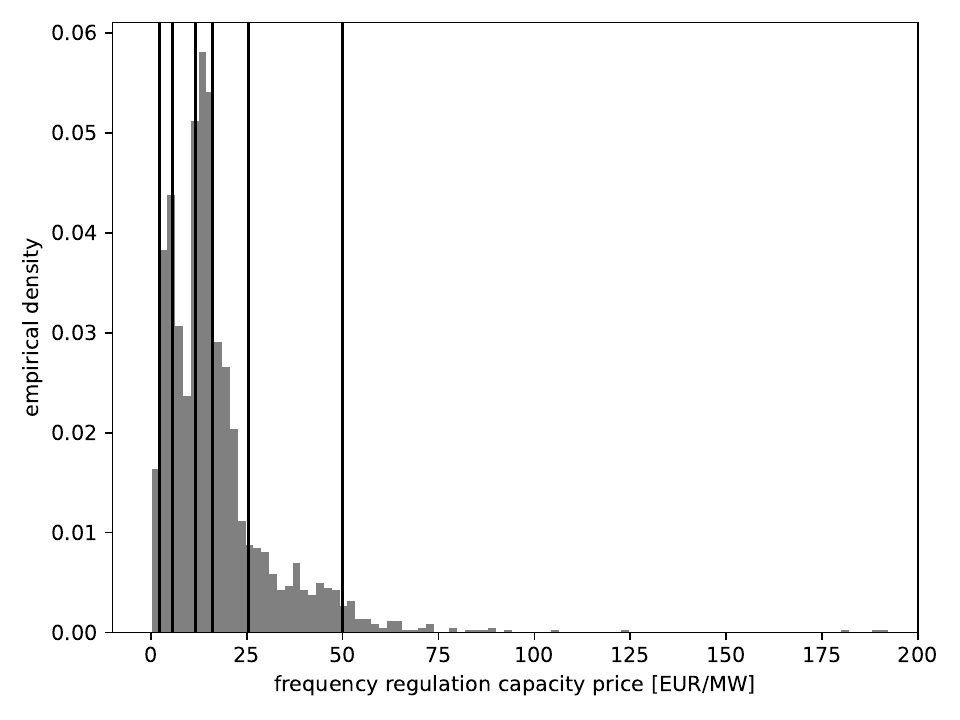}  
    \end{subfigure}
    \caption{Ground truth and representative scenarios for electricity and frequency regulation prices.}
    \label{fig:price_scenarios}
\end{figure}

For generation of the frequency regulation signal scenarios, we use a rank-5 Karhunen-Loève (KL) decomposition which captures \SI{99}{\percent} of the empirical autocovariance of the signal within each bidding interval; see \Cref{fig:frequency_regulation}. The scenarios are constructed from the 5-dimensional second-order Gauss-Hermite quadrature nodes and weights for the KL expansion coefficients. The underlying empirical data and resultant scenarios are shown in \Cref{fig:frequency_regulation}.

\begin{figure}[t]
  \centering
  \includegraphics[width=0.8\linewidth]{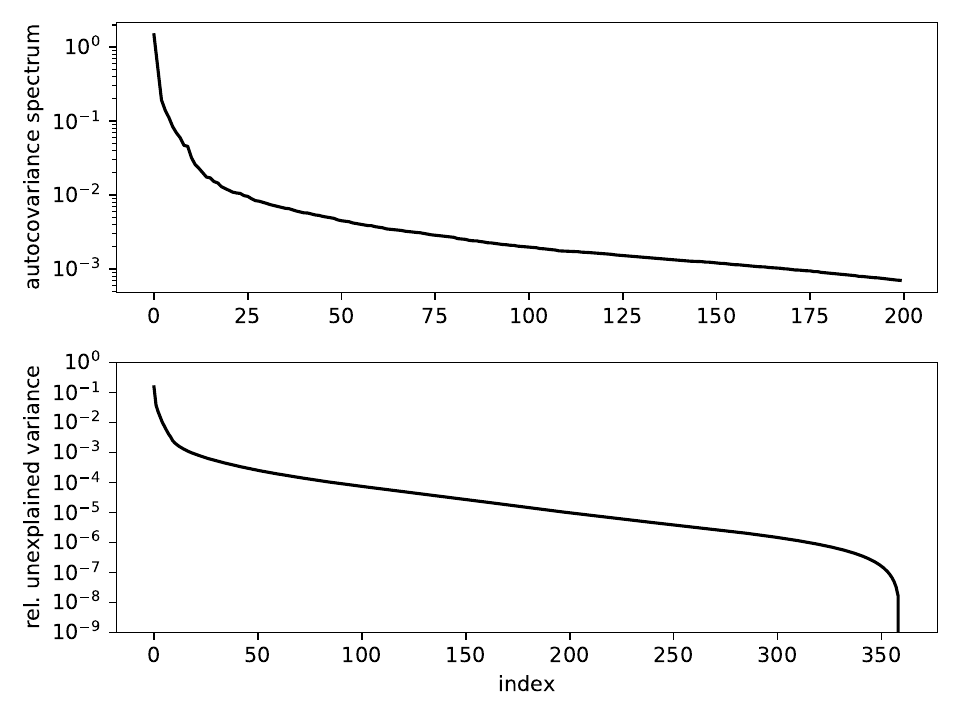}  
  \includegraphics[width=0.8\linewidth]{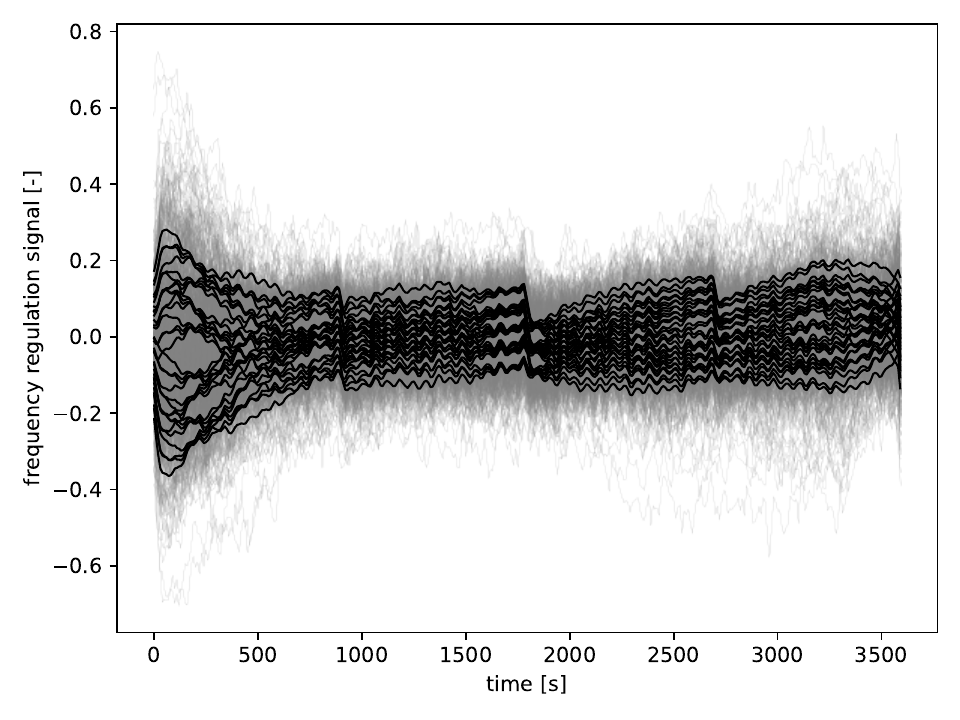}
  \caption{Uncertainty model for frequency regulation signal. Top: Spectrum and normalized unexplained variance of low-rank approximation to the empirical autocovariance matrix of the hourly frequency regulation signals. Bottom: Hourly frequency regulation signal (gray) and representative scenarios (black).}
  \label{fig:frequency_regulation}
\end{figure}

The individual scenarios for frequency regulation signal as well as electricity and frequency regulation capacity prices are combined as a tensor product (implicitly assuming independence) to form a representative set $\mathcal{M}$ of 1138 joint market scenarios considered per stage in the uncertainty model.

\subsection{Closed-Loop Simulations}
In this section, we compare the performance of policies derived via the proposed \gls*{adp} framework against two heuristic benchmark policies. All three considered policies derive online market participation decisions from solving the one-step predictive problem \cref{eq:mpc} ahead of each bidding interval but with different choices for the terminal cost $\varphi$.
The benchmark heuristics deploy a simplified terminal cost, penalizing capacity fade and \gls*{soc} deviation with constant factors. Specifically, we consider the choices
\begin{align}
  \varphi(h, q, \xi) = \alpha (h_{\rm ref} - h), \ \xi \in \text{supp} \, \Xi|\xi_0^0 \label{eq:heuristic_1}
\end{align}
and 
\begin{align}
    \varphi(h, q, \xi) = \alpha (h_{\rm ref}-h) + \mathds{1}_{\lbrace 0.5 \rbrace}(q), \ \xi \in \text{supp} \, \Xi|\xi_0^0 \label{eq:heuristic_2}
\end{align} 
as benchmarks. For both terminal costs, $\alpha$ represents a tunable degradation penalty intended to approximate the marginal value of capacity loss. The convex indicator $\mathds{1}_{\lbrace 0.5\rbrace}(q)$ in Heuristic \cref{eq:heuristic_2} constraints the battery \gls*{soc} in addition to remain at \SI{50}{\percent} at the end of each bidding interval to avoid myopic discharge. The reference battery health state $h_{\rm ref}$ is chosen to coincide with the battery health at the beginning of the bidding interval. In particular Heuristic \cref{eq:heuristic_2} has been demonstrated in \cite{cao2020multiscale} to notably outperform a range of other market participation strategies, including strategies based on coarser battery models and longer look-ahead horizons, in backtests on historical market data. 

In contrast to the above benchmarks, our method derives the terminal cost from point-wise approximate value function and sensitivity evaluations computed offline via \Cref{alg:adp}. \Cref{alg:adp} is applied on a coarse health grid $\mathcal{H}$ consisting of 100 points equidistantly spaced throughout the battery life span, i.e., from \SI{100}{\percent} to \SI{80}{\percent} nominal capacity. Similarly, we use a uniform \gls*{soc} grid $\mathcal{Q}$ covering the feasible range in increments of \SI{5}{\percent}. The battery and market uncertainty models are chosen as described in \Cref{sec:battery_model,sec:uncertainty_model}, respectively. For online deployment, the terminal cost is approximated using the regression problems \cref{eq:regression} on the value function and sensitivity information at the health grid point closest to the current battery health. For the regression problems \cref{eq:regression}, we use the same parametric function approximators,
\begin{align*}
    \widehat{V}_\theta(q) = {\rm d}\widehat{V}_{\theta}(q) =  \begin{bmatrix}
        1&
        q&
        q^2&
        e^{-q}&
        e^{q} 
    \end{bmatrix}\theta,
\end{align*}
alongside mean-squared error loss in both the offline and online phases.

To disentangle the accuracy of the uncertainty and battery model from errors introduced in the proposed value function approximation scheme, we first compare closed-loop simulations under the three different participation policies using the market uncertainty and battery model as ground truths. \Cref{fig:results} compares the cumulative returns attained by the different policies for an exhaustive range of degradation penalties in \cref{eq:heuristic_1} and \cref{eq:heuristic_2}. The proposed \gls*{adp} framework outperforms the heuristics \cref{eq:heuristic_1} and \cref{eq:heuristic_2} by more than $\SI{50}{\percent}$ and $\SI{25}{\percent}$, respectively.

This experiment further underlines the computational merits of the proposed framework. While \Cref{alg:adp} could be executed using merely 100 grid points along the battery health pseudo-time axis, the battery life spans more than 500 days ($12,000$ bidding intervals) in real time under the resultant policy. As a consequence, the value function approximation itself required significantly less computational time than a single closed-loop simulation of the entire battery life. On a standard HPC-compute node with 96 Intel\textsuperscript{\tiny\textregistered} Xeon\textsuperscript{\tiny\textregistered} Platinum 8160 CPU cores and \SI{376}{\giga\byte} of RAM, executing \Cref{alg:adp} required approximately \SI{8}{\hour} in contrast to more than \SI{24}{\hour} needed for evaluation of a single closed-loop simulation. It could therefore be of independent interest to apply the ideas put forward in \Cref{alg:adp} to hyperparameter tuning for heuristics akin to \cref{eq:heuristic_1} and \cref{eq:heuristic_2}.  

\begin{figure*}[t]
  \centering
  \includegraphics[height=.22\linewidth, clip, trim={0 0 73 0}]{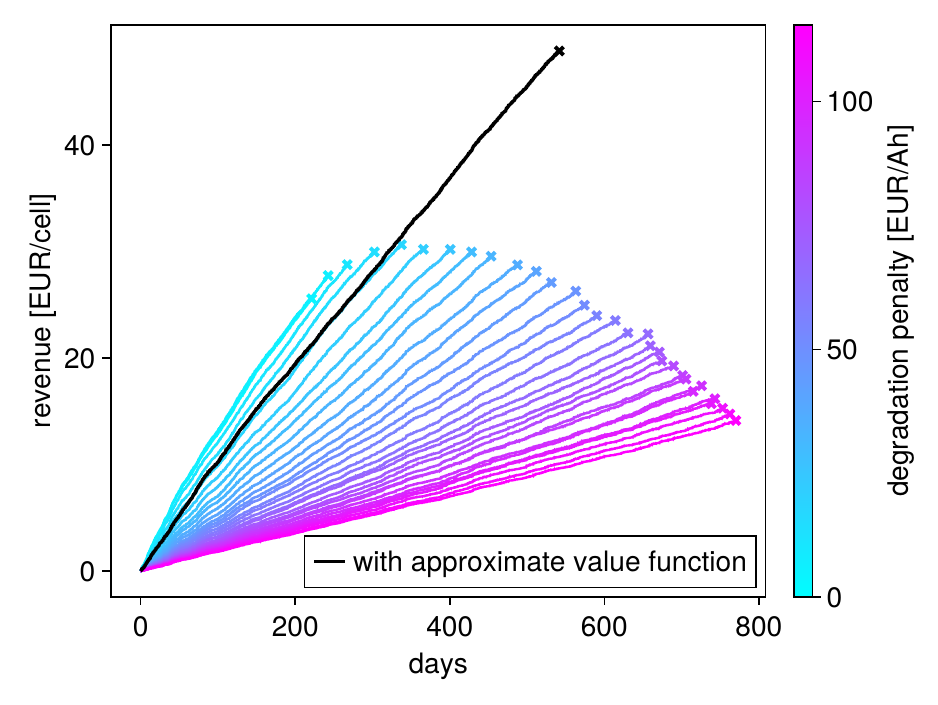}
  \includegraphics[height=.22\linewidth, clip, trim={31 0 73 0}]{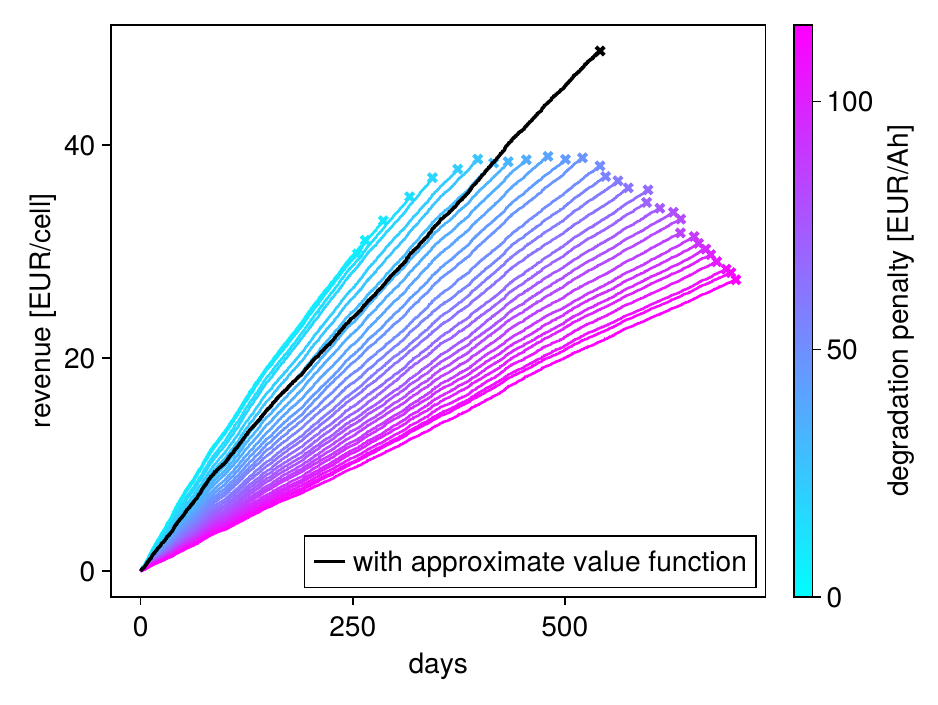}
  \includegraphics[height=.22\linewidth, clip, trim={31 0 73 0}]{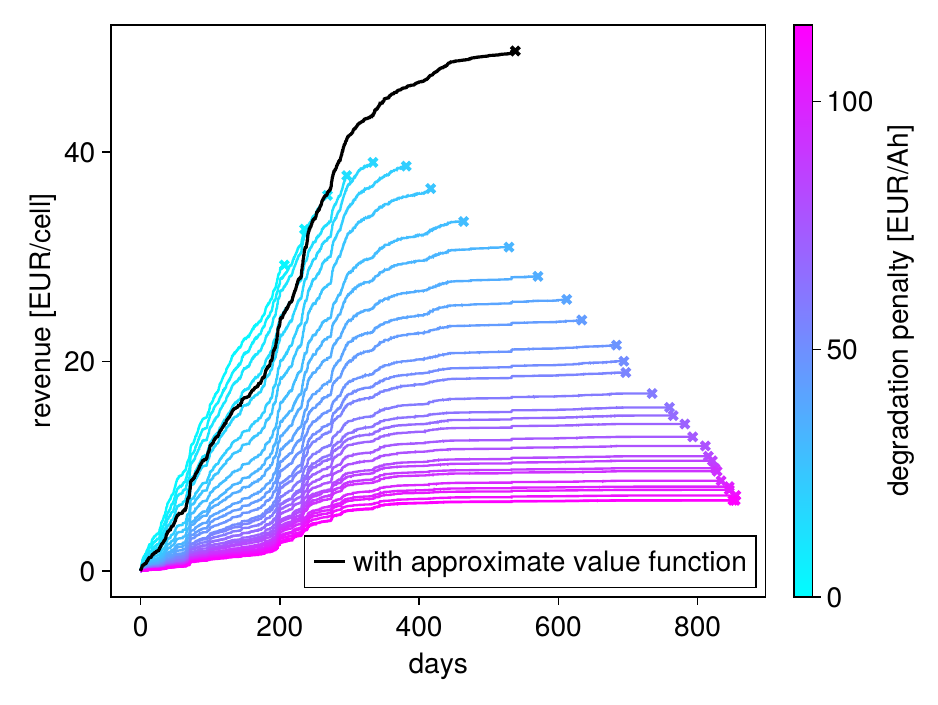}
  \includegraphics[height=.22\linewidth, clip, trim={31 0 0 0}]{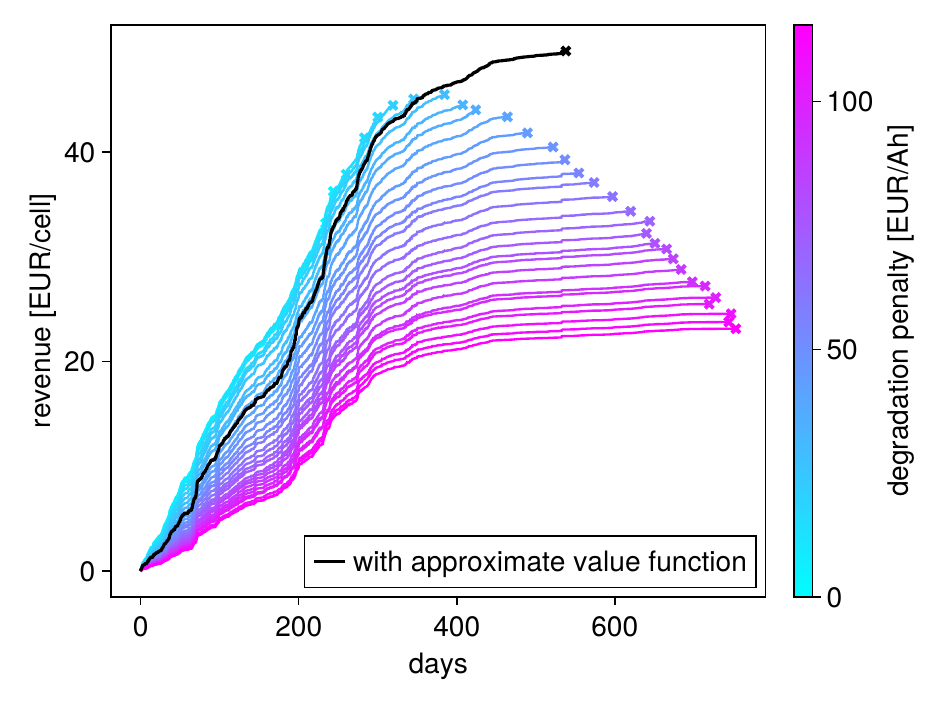}
  \caption{Cumulative returns under approximate value function-informed participation policy versus degradation penalty heuristic.
    The lines end where the battery reaches its end of life. Left to right: closed-loop simulations with ground truth uncertainty model, compared against Heuristic \cref{eq:heuristic_1}; closed-loop simulations, compared against Heuristic \cref{eq:heuristic_2}; backtests on historical market data, compared against Heuristic \cref{eq:heuristic_1}; backtests on historical market data, compared against Heuristic \cref{eq:heuristic_2}.}
  \label{fig:results}
\end{figure*}

\subsection{Backtests on Historical Market Data}
Finally, we evaluate both the benchmark and the proposed participation policies by simulating their closed-loop performance on historical market data from 2022 onward. (We recall that the uncertainty model was constructed using data from 2021 exclusively.) 

\Cref{fig:results} compares the cumulative returns of the proposed participation policy with the heuristics \cref{eq:heuristic_1} and \cref{eq:heuristic_2}. Despite the crude uncertainty model and without any hyperparameter tuning, our proposed participation policy outperforms the benchmark heuristics by more than \SI{25}{\percent} and \SI{10}{\percent}, respectively. In addition, our findings support that the one-step predictive problems \cref{eq:mpc}, which determine online market participation decisions, are indeed reliably tractable for real-time deployment of the derived policy. On the same HPC node used for the offline phase of our \gls*{adp} framework, \SI{99.9}{\percent} of \gls*{mpc} problem instances were solved in under \SI{4.6}{\second}, with an average solve time of \SI{1.6}{\second}. In more than $12,000$ bidding intervals, only a single significant outlier due to numerical issues was observed, requiring a solve time of \SI{479}{\second}.

\section{Conclusions}\label{sec:conclusion}
We presented an \gls*{adp} framework for the design of degradation-aware participation policies for \gls*{bess} in real-time electricity markets. By applying the \gls*{dp} recursion along a pseudo-time axis defined by battery health, our approach naturally leverages the separation of timescales between slow degradation processes and fast operational dynamics. The result is a tractable offline/online computation scheme for making degradation-aware market participation decisions in real time, yet based on high-fidelity physicochemical battery models. 

Our experiments demonstrate that the proposed framework gives rise to policies that effectively balance short-term profits with the long-term impacts of market participation on battery longevity and can lead to superior market performance when compared to common heuristics. These findings support the notion that degradation-aware market participation of \glspl*{bess} may indeed not only be more sustainabile but also increase profitability.

Future work should focus on composing the presented framework with more sophisticated uncertainty models that capture the temporal correlation and prevalent periodicity in electricity grid and market signals.

\printbibliography
\end{document}